\pdfoutput=1
\documentclass[conference]{IEEEtran}
\IEEEoverridecommandlockouts

\usepackage{cite}
\usepackage{amsmath,amssymb,amsfonts}
\usepackage{algorithmic}
\usepackage{graphicx}
\usepackage{textcomp}
\usepackage{xcolor}
\usepackage{url}

\usepackage[textsize=tiny]{todonotes}

\usepackage{physics}
\usepackage{flushend} %
\usepackage{nicematrix}
\usepackage{xfrac}

\usepackage{subcaption}
\usepackage{graphicx}
\usepackage{tikz}
\usepackage{pgfplots}
\usepackage{csvsimple}
\usetikzlibrary{shadows,arrows.meta,backgrounds,calc,chains,matrix,positioning,shapes}
\usetikzlibrary{arrows.meta,backgrounds,calc,chains,matrix,positioning,shapes,shapes.geometric,	shapes.arrows, decorations.pathmorphing, decorations.pathreplacing}
\tikzset{%
	font={\footnotesize},
	vertex/.style={draw,circle,inner sep=0pt,minimum width=0.5cm,minimum height=0.5cm,font=\small, scale=0.9},
	terminal/.style={draw,regular polygon,regular polygon sides=4,inner sep=0pt,minimum width=0.5cm,minimum height=0.5cm,font=\small, scale=1.0},
	zeroterm/.style={below,inner sep=0pt,font=\scriptsize, scale=0.9}
}

\usepackage{amsthm} 

\newtheorem{example}{Example}
\newtheorem{theorem}{Theorem}

\usepackage[a4paper, total={184mm,239mm}]{geometry}
\def\BibTeX{{\rm B\kern-.05em{\sc i\kern-.025em b}\kern-.08em
    T\kern-.1667em\lower.7ex\hbox{E}\kern-.125emX}}
   
\renewcommand{\baselinestretch}{.86} %

\begin{document}

\title{Stochastic Quantum Circuit Simulation\\Using Decision Diagrams} %

\author{\author{
\IEEEauthorblockN{Thomas Grurl\IEEEauthorrefmark{1}\IEEEauthorrefmark{2} \hspace{1cm} Richard Kueng\IEEEauthorrefmark{2} \hspace{1cm} Jürgen Fuß\IEEEauthorrefmark{1} \hspace{1cm} Robert Wille\IEEEauthorrefmark{2}\IEEEauthorrefmark{3}}
\IEEEauthorblockA{\IEEEauthorrefmark{1}Secure Information Systems, University of Applied Sciences Upper Austria, Austria}
\IEEEauthorblockA{\IEEEauthorrefmark{2}Institute for Integrated Circuits, Johannes Kepler University Linz, Austria}
\IEEEauthorblockA{\IEEEauthorrefmark{3}Software Competence Center Hagenberg GmbH (SCCH), Hagenberg, Austria}
\IEEEauthorblockA{\{thomas.grurl, juergen.fuss\}@fh-hagenberg.at \hspace{0.5cm} \{richard.kueng, robert.wille\}@jku.at}
\IEEEauthorblockA{\url{http://iic.jku.at/eda/research/quantum/}}
}
}

\maketitle

\begin{abstract}
Recent years have seen unprecedented advance in the design and control of quantum computers. %
Nonetheless, their applicability is still restricted and access remains expensive. 
Therefore, a substantial amount of quantum algorithms research %
still relies on simulating quantum circuits on classical hardware.
However, due to the sheer complexity of simulating real quantum computers, many simulators unrealistically simplify the problem and instead simulate \emph{perfect} quantum hardware, i.e., they do not consider errors caused by the fragile nature of quantum systems.
\emph{Stochastic quantum simulation} provides a conceptually suitable solution to this problem: physically motivated errors are applied in a probabilistic fashion throughout the simulation.
In this work, we propose to use decision diagrams, as well as concurrent executions, to substantially reduce resource-requirements---which are still daunting---for stochastic quantum circuit simulation. 
Backed up by rigorous theory, %
empirical studies show that this approach allows for 
 a substantially faster and much more scalable simulation for certain quantum circuits.

\end{abstract}

\section{Introduction} 

By utilizing quantum mechanical effects, 
quantum computers promise to solve problems which are intractable for classical computers. Early examples for this are Shor's algorithm~\cite{DBLP:journals/siamcomp/Shor97} for factoring integers or Grover's database search algorithm~\cite{NC:2000}. %
As the research on quantum algorithms gained more and more traction, more quantum algorithms have been found in the areas of chemistry, finance, machine learning, and mathematics\mbox{~\cite{Kassal18681, RebentrostFinance,KerendisQmeans, farhiQAOA}}. 
Alongside the work of quantum software development, there have been unprecedented accomplishments towards the physical realization of quantum hardware. In 2019, Google claimed to have achieved quantum advantage by using a 54-qubit processor to calculate a task in 200 seconds for which they estimate a \mbox{state-of-the-art} supercomputer would require ten thousand years~\cite{Arute2019a}. In the same year, IBM released its \mbox{53-qubit} quantum computer and made it available for commercial use~\cite{ibm53qubit}.
And, recently, IBM announced their roadmap towards launching a \mbox{1212-qubit} processor in 2023~\cite{ibm1023qubitProcessor2020}.

However, recent breakthroughs notwithstanding, quantum computers are still an emerging technology and current quantum processors are limited in reliability and availability. 
Consequently, a considerable amount of research on quantum algorithms still relies on simulating quantum circuits on classical hardware. 
This is an exponentially hard task, almost by definition.
To make matters worse, classical simulation of a \emph{perfect} quantum circuit is, arguably, beside the point. Today's quantum architectures are plagued by frequent errors that are unavoidable given the fragile nature of quantum systems~\cite{Preskill2018quantumcomputingin}. 
Although error mitigation is constantly improving, they are still a dominating factor in quantum computing. 
Therefore, taking those errors into account when simulating quantum circuits is essential in order to understand how an algorithm behaves when executed on real quantum hardware. 

As a theory, quantum mechanics is capable of describing these types of errors and their effect (namely quantum channels and mixed states~\cite{NC:2000}).
However, this general formalism renders
the exponentially hard problem of quantum circuit simulation even harder---accordingly limiting corresponding simulation approaches (see, e.g.,~\cite{forest,qiskit,atos2016,qxSimulator2017,DBLP:journals/corr/WeckerS14,CirqPythonFramework,jones2018quest,DBLP:journals/corr/SmelyanskiySA16,villalonga2019highperformancesimulator,DBLP:conf/iccad/GrurlFW20}).
This is why many quantum circuit simulators simplify the problem and only mimic \emph{perfect} (i.e., error-free) quantum computers (e.g.,~\mbox{\cite{vidal2003efficient,DBLP:journals/tcad/ZulehnerW19,MTG:2006,DBLP:journals/tcad/NiemannWMTD16,Steiger2018projectqopensource}}).

In this work, we consider an alternative approach which avoids making a hard problem unnecessarily harder.  %
Instead, we consider a stochastic error model. That is, we assume that errors occur randomly throughout individual simulation runs. 
Afterwards, we approximate the true effect of errors on a quantum computation by forming empirical averages over multiple simulation runs (Monte Carlo).
This provides a conceptually suitable and mathematically rigorous solution for classical simulation of noisy quantum computations which can easily be implemented on top of existing simulators, such as~\cite{atos2016,qiskit,forest}.

Nonetheless, severe challenges remain. First and foremost, there is the curse of dimensionality: performing a single simulation run requires repeated matrix-vector multiplications of exponential size to appropriately track the effect of quantum operations. To make matters worse, a single stochastic simulation run does not capture error effects appropriately. Instead, a sufficiently large number of independent simulation runs must be conducted to form empirical averages that accurately reflect the true quantum evolution. 
Both factors combined render
 existing solutions severely limited with respect to efficiency and scalability.

In order to overcome these limitations, we propose a solution which (1)~uses decision diagrams to represent states and operations in a more compact fashion and (2)~conducts concurrent executions to accelerate the generation of samples. Evaluations and comparisons to state-of-the-art simulators by IBM and Atos confirm the viability of the proposed solution. %
In fact, for certain circuits, we were able to conduct the respective simulations in a more scalable fashion (i.e,. considering substantially more qubits than before) and a much more efficient fashion (often, several orders of magnitudes faster). %

Our contributions are described in the rest of this paper as follows: Section~\ref{sec:background} reviews quantum computing and the errors that might occur. Section~\ref{sec:stochastic_simulation} discusses how those errors can be simulated using a stochastic approach. In Section~\ref{sec:prob_soluation}, we outline the concept of the proposed solution, which we then evaluate in Section~\ref{sec:evaluation} against two \mbox{state-of-the-art} simulators. Finally, Section~\ref{sec:conclusion} concludes the paper.

\section{Background}
\label{sec:background}

In order to keep this work self-contained, this section reviews the basic concepts of quantum computing as well as error effects. We refer the interested reader to standard textbooks, e.g.,~\cite{NC:2000,watrous_2018}, for a more thorough introduction.

\subsection{Quantum Computing}
\label{subsec:quantum_computing}
In the classical world, the basic unit of information is a bit, which can either assume the state 0 or 1. In the quantum world, the smallest unit of information is called a \emph{quantum bit} or \emph{qubit}. Like a classical bit, a qubit can assume the states 0 and 1, which are called \emph{basis states} and---using Dirac notation---are written as $\ket{0}$ and $\ket{1}$. Additionally, a qubit can also assume an almost arbitrary combination of the two basis states, which is then called a \emph{superposition}. More precisely, the state of the qubit $\ket{\psi}$ is written as $\ket{\psi} = \alpha_0 \cdot \ket{0} + \alpha_1 \cdot \ket{1}$ with $\alpha_0, \alpha_1 \in \mathbb{C}$ such that \mbox{$\abs{\alpha_0}^2 + \abs{\alpha_1}^2 = 1$}. The values $\alpha_0, \alpha_1$ are called \emph{amplitudes} and describe how strongly the qubit is related to each of the basis states. Measuring a qubit yields $\ket{0}$ ($\ket{1}$) with probability $\abs{\alpha_0}^2$ ($\abs{\alpha_1}^2$). By measuring the qubit, any existing superposition is destroyed and the state of the qubit collapses to the measured basis state.

Quantum states containing more than one qubit are often called \emph{quantum registers} and the concepts above can be extended to describe such systems as well. An \mbox{$n$~qubit} register can assume $N = 2^{n}$ basis states and is described by $N$ amplitudes $\alpha_0, \alpha_1, \dots \alpha_{N-1}$, which must satisfy the normalization constraint $\sum_{i\in \{0,1\}^n} |\alpha_i|^2 = 1$. 
Usually quantum states are shortened to state vectors containing only the amplitudes, e.g., 
$ \begin{bNiceMatrix}[small] \alpha_{00}&\alpha_{01}&\alpha_{10}&\alpha_{11}\end{bNiceMatrix}^\top$ for $n=2$ qubits.  %

\begin{example}
\label{exp:state_vec_rep}
Consider the 2-qubit quantum register
\vspace*{-1mm}
\[
\ket{\psi} = \frac{1}{\sqrt{2}} \cdot \ket{00} + 0 \cdot \ket{01} + \frac{1}{\sqrt{2}} \cdot \ket{10} + 0 \cdot \ket{11},
\]
\vspace*{-1mm}
which is represented by the state vector $\begin{bNiceMatrix}[small] \frac{1}{\sqrt{2}}&0&\frac{1}{\sqrt{2}}&0\end{bNiceMatrix}^\top$.
This is a valid state, since it satisfies the normalization constraint
$\abs{{\frac{1}{\sqrt{2}}}}^2 + 0^2 + \abs{{\frac{1}{\sqrt{2}}}}^2 + 0^2 = 1$.
Measuring the system yields either $ \ket{00} $ or $ \ket{10} $---both with probability $\abs{{\frac{1}{\sqrt{2}}}}^2= {\frac{1}{2}}$. 
Note that the leftmost qubit is in a superposition and equally strongly related to $\ket{0}$ and $\ket{1}$, while the other qubit is in the basis state $\ket{0}$.
\end{example}

Quantum states can be manipulated using quantum operations. With the exception of the measurement operation, all quantum operations are inherently reversible and represented by unitary matrices, i.e., square matrices whose inverse is their conjugate transpose. Important 1-qubit operations are \mbox{$\textup{H}=\sfrac{1}{\sqrt{2}}\begin{bNiceMatrix}[r][small]1&1\\1&-1\end{bNiceMatrix}$} (transforming a basis state into a superposition), 
$\textup{X}=\begin{bNiceMatrix}[r][small]0&1\\1&0\end{bNiceMatrix}$ (the quantum equivalent of the NOT operation), 
$\textup{Z}=\begin{bNiceMatrix}[r][small]1&0\\0&-1\end{bNiceMatrix}$ (flipping the phase of a qubit), as well as the combination of both $Y=iXZ=\begin{bNiceMatrix}[r][small]0&-i\\i&0\end{bNiceMatrix}$.
Another ``operation'' is the identity operation given by $\textup{I}=\begin{bNiceMatrix}[r][small]1&0\\0&1\end{bNiceMatrix}$. It simply leaves a state unchanged and is relevant in the context of simulating errors. There are also 2-qubit operations. An important example is the controlled-X (also known as CNOT) operation, which negates the state of a qubit, if the chosen control qubit is $\ket{1}$. 
Applying an operation to a state can be done by matrix-vector multiplication.

\begin{example}
\label{exp:matrxi_vector_mul}
Consider again the 2-qubit register $\ket{\psi}$ from Example~\ref{exp:state_vec_rep}. Applying a CNOT operation to $\ket{\psi}$, which negates the amplitude of the second qubit if the first qubit is set to~$\ket{1}$, is given by
\vspace*{-1mm}
\begin{align*}
\underbrace{\begin{bmatrix}1&0&0&0\\0&1&0&0\\0&0&0&1\\0&0&1&0\end{bmatrix}}_{\textup{CNOT}}  \cdot  \underbrace{\begin{bmatrix}\frac{1}{\sqrt{2}}\\0\\\frac{1}{\sqrt{2}}\\0\end{bmatrix}}_{\mathit{\ket{\psi}}} = \underbrace{\begin{bmatrix}\frac{1}{\sqrt{2}}\\0\\0\\\frac{1}{\sqrt{2}}\end{bmatrix}}_{\mathit{\ket{\psi^\prime}}}.
\end{align*}
\vspace*{-1mm}
Measuring $\ket{\psi^\prime}$ either yields $ \ket{00} $ or $ \ket{11} $, each with probability~${\sfrac{1}{2}}$. Note that the measurement outcome of one qubit affects the other one as well---an essential concept in quantum computing known as \emph{entanglement}.
\end{example}

\subsection{Errors in Quantum Computing} %
\label{subsec:back_errors}
The formalism presented above can be used to describe how perfect quantum computers behave. However, due to the fragile nature of quantum systems, real quantum computers are prone to errors. 
These errors can be classified into two categories~\cite{Tannu2018}: \emph{Gate errors} (also known as operational errors) and \emph{coherence errors} (also known as retention errors).

\subsubsection{Gate Errors} are introduced by executed operations~\cite{Tannu2018}. They occur since quantum computers are mechanical constructions that do not always apply operations perfectly. Instead, the operation may be not executed at all, or in a (slightly) modified fashion. Since gate errors are highly specific for each quantum computer and even vary for qubits within the quantum computer, they are often approximated using depolarization errors~\cite{qxSimulator2017, qiskit}. The depolarization error describes that a qubit is set to a completely random state~\cite{NC:2000}.
For publicly available quantum computers from IBM, the %
error probabilities are in the order of $10^{-3}$ to $10^{-2}$~\cite{ibmErrorRate2019}.

\subsubsection{Coherence Errors} occur due to the fragile nature of quantum systems (qubits). 
In practice, this leads to the problem that they can hold information for a limited time only. There are two types of coherence errors that may appear~\cite{Tannu2018}: 
\begin{itemize}
\item A qubit in a \mbox{high-energy} state ($\ket{1}$) tends to relax into a low energy state ($\ket{0}$). That is, after a certain amount of time, qubits in a quantum system eventually decay to~$\ket{0}$. This error is called \emph{amplitude damping error} or \emph{T1~error}. 

\item In addition to that, when a qubit interacts with the environment, a phase flip effect might occur. This leads to an error called \emph{phase flip error} or \emph{T2~error}.
\end{itemize}

Developments in the physical realization of quantum computers (e.g., in~\cite{Devoret1169, Kelly2018}) show significant improvements in the coherence times %
of qubits---improving the ``lifetime'' of qubits before decaying to~$\ket{0}$ and reducing the frequency of phase flip errors, respectively. Nevertheless, the errors are still a significant aspect in all quantum computations and, hence, should also be considered during simulation.

Error effects can be viewed as (unwanted) operations on the state. 
However, while ideal quantum operations are deterministic, errors can have an additional degree of randomness. For instance, a \mbox{1-qubit} amplitude damping channel may fire (with probability~$p$) or it may do nothing (with probability~$1-p$).
The outcome of an erroneous quantum computation cannot be described as a single state $|\psi \rangle$ anymore. Instead, it is described by an \emph{ensemble} of possible outcomes $\left\{(p_i, |\psi_i \rangle) \right\}$. Here, the states~$|\psi_i \rangle$ label potential outcomes while each weight $p_i$  describes the probability with which outcome $|\psi_i \rangle$ occurs ($p_i \geq 0$ and $\sum_i p_i=1$). 

\begin{example}\label{exp:errors}
Consider again the 2-qubit state \mbox{$\ket{\psi^\prime} = \frac{1}{\sqrt{2}}(\ket{00} + \ket{11})$} from Example~\ref{exp:matrxi_vector_mul}. Suppose that this state might be affected by a gate error in the first qubit only, depolarizing it. 
With probability $1-p$, nothing happens and the state remains unchanged. With probability $p$, the first qubit becomes depolarized. We can capture this effect by either applying \mbox{I}, \mbox{X}, \mbox{Y}, or \mbox{Z}---each with probability $\tfrac{p}{4}$. This produces an ensemble (or mixture)
$
\{(1-p,\tfrac{1}{\sqrt{2}}(\ket{00} + \ket{11}), (\tfrac{p}{4},\tfrac{1}{\sqrt{2}}(\ket{00} + \ket{11}),(\tfrac{p}{4},(\tfrac{1}{\sqrt{2}}(\ket{01} + \ket{10}), (\tfrac{p}{4},(\tfrac{i}{\sqrt{2}}(-\ket{01} + \ket{10}), (\tfrac{p}{4},\tfrac{1}{\sqrt{2}}(\ket{00} - \ket{11})\}
$
which cannot be represented by a single 2-qubit state. %
\end{example}

\section{Stochastic Quantum Circuit Simulation}
\label{sec:stochastic_simulation}

In order to conduct quantum circuit simulation, the concepts described in Section~\ref{sec:background} need to be emulated on a classical machine. Conceptually, this can be conducted in a straightforward fashion: State vectors and operation matrices are first represented in the form of 1-dimensional and 2-dimensional arrays, respectively. Then, the application of quantum states is handled by applying \mbox{matrix-vector} multiplication as illustrated in Example~\ref{exp:matrxi_vector_mul} above. 

However, the problem of this approach is that the representation of both, quantum states and quantum operations require exponentially large vectors and matrices---rendering quantum circuit simulation very complex. Moreover, error effects occur only by chance, i.e., they are randomly applied depending on the hardware model of the simulated quantum computer and cannot simply be considered in a \mbox{pre-defined} fashion (such as operations). Therefore, instead of one possible final state, simulation with errors produces a range of possible output states, depending on the applied error effects (as illustrated in Example~\ref{exp:errors} above).

This phenomenon is well known and may be captured by a rigorous mathematical formalism: quantum channels and mixed states~\cite{NC:2000,watrous_2018}. However, this formalism necessarily amplifies the curse of dimensionality: mixed states correspond to $2^n \times 2^n$ matrices and keeping track of them renders an exponentially large problem even harder.

Stochastic quantum simulation, on the other hand, avoids this further increase in simulation complexity by sacrificing deterministic descriptions. The key idea is to imitate error effects in a real quantum computer. That is, whenever the quantum computer might make an error during its calculation with some probability~$p$, we mimic the effect of this error during the simulation with probability~$p$ and leave the state untouched with probability~$1-p$. By simulating in such a way, we generate \emph{one possible} final state $|\tilde{\psi} \rangle$ that is sampled from the actual ensemble $\left\{ (p_i, |\psi_i \rangle) \right\}$:
$
|\tilde{\psi} \rangle = |\psi_i \rangle$ with probability~$p_i$.
Sampling access opens the door for stochastic \mbox{(Monte-Carlo)} approximation: simply approximate the true distribution by forming empirical averages of sampled output states.

Stochastic approximation is particularly well suited for directly and accurately learning interesting properties of the final state (distribution) without the need of keeping track of the complete distribution. 
In quantum computing, many interesting properties can be described in terms of quadratic functions in the state vector, i.e., $o_l = |\langle \omega_l |\psi \rangle|^2$. Prominent examples are the fidelity with another state, as well as the outcome probability of a computational basis measurement. 
For a probabilistic state mixture $\left\{(p_i,|\psi_i \rangle)\right\}$ such a quadratic property becomes
\begin{equation}
o_l = \sum_i p_i  \left| \langle \omega_l| \psi_i \rangle \right|^2 
\label{eq:quadratic-property}
\end{equation}
and can be approximated by an empirical average over $M$  samples $|\tilde{\psi}_j \rangle$ from this distribution:
\begin{equation*}
\hat{o_l}= \tfrac{1}{M} \sum_{j=1}^M \left|\langle \omega_l| \tilde{\psi}_j \rangle \right|^2 \quad \text{(Monte Carlo)}.
\end{equation*}

Moreover, the same collection of samples $\left\{|\tilde{\psi}_1\rangle,\ldots,|\psi_M \rangle \right\}$ can be used to estimate many quadratic properties at once. 

\begin{theorem} \label{theorem_stoch}
Fix a collection of $L$ (arbitrary) quadratic properties \eqref{eq:quadratic-property}, as well as $\epsilon \in (0,1)$ (accuracy) and $\delta \in (0,1)$ (confidence). 
Then, $M = \log (2L/\delta)/(2\epsilon)^2$ state samples suffice to accurately approximate \emph{all} target properties with high confidence: $ \max_l |\hat{o}_l - o_l| \leq \epsilon$ with probability at least $1-\delta$.
\end{theorem}

\begin{proof}
Fix a target property $o_l = \sum_i p_i |\langle \omega_l| \psi_i \rangle |^2$. 
Conducting a single stochastic run yields the correct property in expectation, i.e.,\ 
$\mathbb{E} |\langle \omega_l |\tilde{\psi}_j \rangle |^2
= o_l$.
Standard concentration inequalities, like Hoeffding, imply \mbox{$\mathrm{Pr} \left[ \left|o_l - \hat{o}_l \right| \geq \epsilon \right] \leq 2 \mathrm{e}^{-2M \epsilon^2}$}.
The claim follows from taking a union bound over all $L$ target approximations and inserting the advertised value of $M$.
\end{proof}
\vspace*{-1mm}
As is typical of Monte Carlo, the required number of samples~$M$ scales inverse quadratically in the desired accuracy~$\epsilon$. %
More importantly and interestingly, $M$ only depends logarithmically on the number $L$ of target properties and is independent of the actual system size. 
 This logarithmic suppression can help to combat the curse of dimensionality. For instance, only roughly $n/\epsilon^2$ samples suffice to $\epsilon$-approximate all $N=2^n$ outcome probabilities of the underlying state distribution.

Overall, stochastic quantum circuit simulation allows to avoid the increase of complexity from 
$2^n$-vectors to \mbox{$2^n\times 2^n$- matrices}. However, the challenge remains to produce and process samples~$|\tilde{\psi}_i \rangle$ (which still remain exponential in size). State-of-the-art quantum circuit simulators like~\cite{atos2016,qiskit,forest} still severely suffer from the remaining exponential complexity.

\section{Proposed Solution}\label{sec:prob_soluation}

In this section, we present a solution that addresses the problems that still exists in stochastic quantum circuit simulation. To this end, we first briefly introduce the main ideas of our solution; followed by more detailed descriptions afterwards. Section~\ref{sec:evaluation} eventually shows that the concepts introduced here have a substantial impact on the performance and the scalability of stochastic quantum circuit simulation.
\subsection{General Ideas}

Stochastic quantum circuit simulation suffers from the fact that (1)~the underlying concepts require exponentially large representations of vectors and matrices and (2)~that, in order to determine accurate predictions (see Theorem~\ref{theorem_stoch}), a sufficient number of simulation runs (with these exponential representations) need to be conducted---posing severe challenges with respect to memory and runtime. In this work, we are addressing these challenges with the following two key ideas:
\begin{itemize}

\item \emph{Use Decision Diagrams for Individual Simulation Runs}: In the conventional realm, decision diagrams (such as proposed, e.g., in~\cite{Bry:86,Min:93,DST+:94b}) 
have found great application to tackle several (exponentially hard) problems. Even in the quantum realm, first approaches successfully exploiting them have been reported (see, e.g.,~\cite{VRMH:2003,abdollahi2006analysis,WLTK:2008,DBLP:journals/tcad/NiemannWMTD16,DBLP:conf/date/ZulehnerW19,DBLP:journals/tcad/ZulehnerW19}). 
We propose to use these promises to tackle the challenge of exponential complexity in individual simulation runs.

\item \emph{Exploit Concurrency Across Different Simulation Runs}: 
In stochastic quantum simulation, interesting properties are obtained by empirically averaging over a sufficient number of \emph{independent} simulation runs. 
This setup facilitates concurrent implementation: use different cores for computing different simulation runs. Since accurate predictions are contingent on empirically
averaging many samples (see Theorem~\ref{theorem_stoch}), the potential of
concurrent protocol execution on multi-core architectures
is enormous.
\end{itemize}

These two ideas turn out to complement each other nicely. 
Concurrency is a default feature of the proposed high-level solution (Monte Carlo), while decision diagrams provide a tractable way for executing individual simulation runs.
In the remainder of this section, details of these ideas are described and illustrated.

\subsection{Using Decision Diagrams for Individual Simulation Runs}
The general idea of decision \mbox{diagram-based} quantum circuit simulation is about uncovering and exploiting redundancies in the representation of states and operations. Doing so can result in potentially very compact representations, which in turn allows simulating quantum circuits that cannot be tackled using other simulation approaches anymore. 

Representing a state vector as a decision diagram revolves around recursively splitting the vector into equally sized \mbox{sub-vectors}, until the \mbox{sub-vectors} only contain a single element. More precisely, consider a quantum register $q_0, q_1,\dots,q_{n-1}$ composed of $n$ qubits, where $q_0$ represents the most significant qubit. The first $2^{n-1}$ entries of the corresponding state vector would then represent amplitudes for basis states where $q_0$ is $\ket{0}$, while the remaining $2^{n-1}$ entries would represent amplitudes where $q_0$ is~$\ket{1}$. This is represented in a decision diagram by a node labeled~$q_0$ with a left (right) successor which points to a node that represents the sub-vector with amplitudes for basis states with $q_0$ assigned $\ket{0}$~($\ket{1}$). This process is repeated recursively until \mbox{sub-vectors} of size 1 (i.e., complex numbers) result. 

During this process, equivalent \mbox{sub-vectors} are represented by the same node---reducing the overall size of the decision diagram. Furthermore, instead of having distinct terminal nodes for all amplitudes, edge weights are used to store common factors of the amplitudes---leading to even more compaction. Reconstructing the amplitude of a specific state can be done by multiplying the edge weights along the corresponding path.

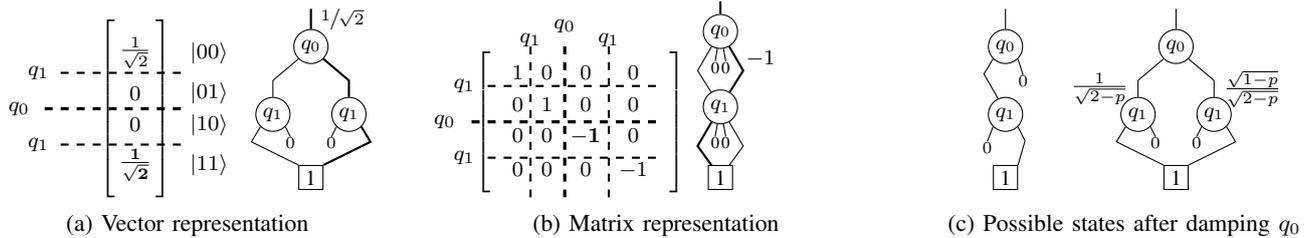
\begin{figure*}[t]
	\begin{subfigure}[b]{0.33\linewidth}
		\centering
		\begin{tikzpicture}
		\matrix[matrix of math nodes, left delimiter={[},right delimiter={]}, inner xsep=0] (vector) {
			\frac{1}{\sqrt{2}}\\				
			0\\
			0\\
			\mathbf{\frac{1}{\sqrt{2}}}\\				
		};			
		\begin{scope}[on background layer, black]	
		\node[right=0.6cm of vector-1-1.center] {$\ket{00}$};
		\node[right=0.6cm of vector-2-1.center] {$\ket{01}$};
		\node[right=0.6cm of vector-3-1.center] {$\ket{10}$};
		\node[right=0.6cm of vector-4-1.center] {$\ket{11}$};
		
		\draw[black,-,dashed, very thick,shorten <= -0.6cm] ($(vector-2-1)!0.5!(vector-3-1)$) -- ++(-1.25,0) node[anchor=east] {\(q_0\)};
		
		\draw[black,-,dashed, thick,shorten <= -0.6cm] ($(vector-1-1)!0.5!(vector-2-1)$) -- ++(-1,0) node[anchor=east] {\(q_1\)};
		\draw[black,-,dashed, thick,shorten <= -0.6cm] ($(vector-3-1)!0.5!(vector-4-1)$) -- ++(-1,0) node[anchor=east] {\(q_1\)};
		
		\end{scope}
		\end{tikzpicture}
		\begin{tikzpicture}	
		\matrix[matrix of nodes,ampersand replacement=\&,every node/.style=vertex,column sep={0.5cm,between origins},row sep={0.9cm,between origins}] (qmdd2) {
			\& \node (m1) {$q_0$}; \& \\
			\node (m2a) {$q_1$}; \& \& \node (m2b) {$q_1$}; \\
			\& \node[terminal] (t3) {1}; \& \\
		};
		
		\draw[thick] ($(m1)+(0,0.5cm)$) -- (m1) node[right, midway]{$\sfrac{1}{\sqrt{2}}$};
		
		\draw (m1.-135) -- ($(m1)!0.5!($(m2a)!0.5!(m2b)$) + (-5mm,0)$) -- (m2a.90);
		\draw[thick] (m1.-45) -- ($(m1)!0.5!($(m2a)!0.5!(m2b)$) + (5mm,0)$) -- (m2b.90);
		
		\draw (m2a.-135) -- ($(m2a.-135) - (1.2mm,2.0mm)$) -- (t3.135);
		\draw (m2a.-45) -- ($(m2a.-45)!0.2!(t3) + (0mm,0)$)  node[zeroterm] {$0$};;
		
		\draw (m2b.-135) -- ($(m2b.-135)!0.2!(t3) + (0mm,0)$)  node[zeroterm] {$0$};;
		\draw[thick] (m2b.-45) -- ($(m2b.-45) - (-1.2mm,2.0mm)$) -- (t3.45);		
		\end{tikzpicture}
		\caption{Vector representation}
		\label{fig:statevectordd}
	    \end{subfigure}\hfill
	    \begin{subfigure}[b]{0.33\linewidth}
		\begin{tikzpicture}
		\matrix[matrix of math nodes, left delimiter={[},right delimiter={]}] (matrix) {
			1 & 0 & 0 & 0\\				
			0 & 1 & 0 & 0\\		
			0 & 0 & \mathbf{-1} & 0\\		
			0 & 0 & 0 & -1\\					
		};			
		\begin{scope}[on background layer, black]	
		
		\draw [black,-,dashed, thick,shorten <= 0cm] (0.38,1) -- (0.38,-1)  ++(0,2.3) node[anchor=north] {\(q_1\)};
\draw [black,-,dashed, very thick,shorten <= 0cm] (-0.20,1) -- (-0.20,-1)  ++(0,2.5) node[anchor=north] {\(q_0\)};			
		
		\draw [black,-,dashed, thick,shorten <= 0 cm] (-0.65,1) -- (-0.65,-1)  ++(0,2.3) node[anchor=north] {\(q_1\)};

		\draw[black,-,dashed, very thick,shorten <= -1cm] ($(vector-2-1)!0.5!(vector-3-1)$) -- ++(-1.50,0) node[anchor=east] {\(q_0\)};
		
		\draw[black,-,dashed, thick,shorten <= -1cm] ($(vector-1-1)!0.5!(vector-2-1)$) -- ++(-1.25,0) node[anchor=east] {\(q_1\)};
		\draw[black,-,dashed, thick,shorten <= -1cm] ($(vector-3-1)!0.5!(vector-4-1)$) -- ++(-1.25,0) node[anchor=east] {\(q_1\)};		
		\end{scope}
		\end{tikzpicture}
		\begin{tikzpicture}
		\matrix[matrix of nodes,ampersand replacement=\&,every node/.style={vertex},column sep={1cm,between origins},row sep={1cm,between origins}] (qmdd) {
			\node (n1) {$q_0$}; \\
			\node (n2) {$q_1$}; \\
			\node[terminal, outer sep=0pt] (t){1};\\
		};
		\draw (n1) -- ++(240:0.6cm) -- (n2);
		\draw (n1) -- ++(260:0.4cm) node[zeroterm]{0};
		\draw (n1) -- ++(280:0.4cm) node[zeroterm]{0};
		\draw (n1)[thick] -- ++(300:0.6cm) node[right,midway] {$-1$} -- (n2);
		
		\draw (n2)[thick] -- ++(240:0.6cm) -- (t);
		\draw (n2) -- ++(260:0.4cm) node[zeroterm]{0};
		\draw (n2) -- ++(280:0.4cm) node[zeroterm]{0};
		\draw (n2) -- ++(300:0.6cm) -- (t);
		\draw[thick] ($(n1)+(0,0.4cm)$) -- (n1);
		\end{tikzpicture}
		\caption{Matrix representation}
		\label{fig:operationmatrix}
	    \end{subfigure}\hfill
		\begin{subfigure}[b]{0.33\linewidth}
		\centering
		\begin{tikzpicture}	
		\matrix[matrix of nodes,ampersand replacement=\&,every node/.style=vertex,column sep={0.5cm,between origins},row sep={0.9cm,between origins}] (qmdd2) {
			\& \node (m1) {$q_0$}; \& \\
			\& \node (m2a) {$q_1$}; \& \\
			\& \node[terminal] (t3) {1}; \& \\
		};
		
		\draw ($(m1)+(0,0.5cm)$) -- (m1);
		
		\draw (m1.-135)  -- ($(m1.-135) - (1.2mm,2.0mm)$) -- (m2a.115);
		\draw (m1.-45) -- ($(m1.-10)!0.2!(t3) + (0.5mm,0)$)  node[zeroterm] {$0$};;
		
		\draw (m2a.-135) -- ($(m2a.-170)!0.3!(t3) + (-0.8mm,-0.8mm)$)  node[zeroterm] {$0$};;
		\draw (m2a.-45) -- ($(m2a.-25)!0.3!(t3) + (1.1mm,-0.3mm)$) -- (t3.45);		
		\end{tikzpicture}			
		\begin{tikzpicture}	
		\matrix[matrix of nodes,ampersand replacement=\&,every node/.style=vertex,column sep={0.5cm,between origins},row sep={0.9cm,between origins}] (qmdd2) {
			\& \node (m1) {$q_0$}; \& \\
			\node (m2a) {$q_1$}; \& \& \node (m2b) {$q_1$}; \\
			\& \node[terminal] (t3) {1}; \& \\
		};
		
		\draw ($(m1)+(0,0.5cm)$) -- (m1);
		
		\draw (m1.-135) -- ($(m1)!0.5!($(m2a)!0.5!(m2b)$) + (-5mm,0)$) -- (m2a.90) node[left, midway]{$\frac{1}{\sqrt{2-p}}$};
		\draw (m1.-45) -- ($(m1)!0.5!($(m2a)!0.5!(m2b)$) + (5mm,0)$) -- (m2b.90)node[right, midway]{$\frac{\sqrt{1-p}}{\sqrt{2-p}}$};
		
		\draw (m2a.-135) -- ($(m2a.-135) - (1.2mm,2.0mm)$) -- (t3.135);
		\draw (m2a.-45) -- ($(m2a.-45)!0.2!(t3) + (0mm,0)$)  node[zeroterm] {$0$};;
		
		\draw (m2b.-135) -- ($(m2b.-135)!0.2!(t3) + (0mm,0)$)  node[zeroterm] {$0$};;
		\draw (m2b.-45) -- ($(m2b.-45) - (-1.2mm,2.0mm)$) -- (t3.45);		
		\end{tikzpicture}		
		\caption{Possible states after damping $q_0$}
		\label{fig:applyDepolToDD}
	\end{subfigure}	
	\caption{Decision diagram representation of states}
	\label{fig:dd_rep_examples}
\end{figure*}

\begin{example}
\label{exp:vectorToDD}
In Fig.~\ref{fig:statevectordd}, the state vector $\ket{\psi^\prime}$ from Example~\ref{exp:matrxi_vector_mul} is represented as both, vector and decision diagram\footnote{In order to aid the readability of the decision diagram, edge weights of~1 are omitted. Additionally, nodes with an incoming edge weight of 0 are represented as 0-stubs---indicating that amplitudes of all possible states represented by this part of the decision diagram are zero.}. The annotations of the vector representation indicate how it is decomposed for the decision diagram representation.
In order to reconstruct an amplitude from the decision diagram, the edge weights of the corresponding path must be multiplied. For example, the amplitude of the state $\ket{11}$ (represented by the bold path in Fig.~\ref{fig:statevectordd}) can be reconstructed by multiplying the edge weights of the root edge (${\frac{1}{\sqrt{2}}}$) with the right edge of~$q_0$ ($1$), as well as the right edge of $q_1$ ($1$), i.e., ${\frac{1}{\sqrt{2}}} \cdot 1 \cdot 1 = {\frac{1}{\sqrt{2}}}$.
\end{example}

Matrix representations of quantum operations are represented as decision diagrams in a similar way. However, due to the square nature of matrices, they are split into four equally sized sub-parts. These parts are represented in a decision diagram by a node with four successor edges. The first one representing the upper left, the second the upper right, the third the lower left, and the fourth the lower right \mbox{sub-matrix}. The remaining decomposition steps are analogous to the case of a vector described above. 

\begin{example}
Fig.~\ref{fig:operationmatrix} provides a decision diagram representation for a Z-operation applied to the first qubit of a 2-qubit register. The annotations in the matrix representation indicate how the matrix is decomposed into the resulting decision diagram. The matrix entry highlighted bold in Fig.~\ref{fig:operationmatrix} can be reconstructed by multiplying the edge weights of the root edge $1$ with the first edge of $q_0$ ($-1$), as well as the first edge of $q_1$ ($1$), i.e., ${1 \cdot -1 \cdot 1 = -1}$. 
\end{example}

Using these representations, operations such as \mbox{matrix-vector} multiplication can be executed so that, simulation of quantum circuits can be conducted. However, similar to the vector and matrix representation, the multiplication must also be decomposed with respect to the most significant qubit. 

More precisely, consider a quantum register composed of $n$ qubits given by $\ket{\phi} = q_0, q_1, \ldots, q_{n-1}$ (where $q_0$ represents the most significant qubit) and a unitary quantum operation $U$ of size $2^n \times 2^n$. To multiply the operation $U$ onto the state~$\ket{\phi}$, they are split into two (in the case of the state vector) and four (in the case of the operation) equally sized parts. This leads to two sub-vectors of size $2^{n-1}$ and four sub-matrices of size $2^{n-1} \times 2^{n-1}$. This represents the modifications of $U$ onto~$q_0$ and is accordingly represented by a top node labeled~$q_0$, with two successor nodes. Similar to the vector and matrix decomposition, this process is recursively repeated until vectors of size~$2$ and matrices of size $2 \times 2$ remain, which are multiplied. From the resulting new amplitudes, the new edge weights are calculated and equivalent sub-vectors are represented by the same node.
Multiplications %
therefore mainly involves recursive traversals of the involved decision diagrams.

Finally, in order to consider error effects, we basically can re-use the concepts from above, i.e., we view error effects as operations, which are applied to the state with some probability. The outcome of such an erroneous operation is an \emph{ensemble} of possible outcomes $\left\{(p_i, |\psi_i \rangle) \right\}$. Gate errors causing depolarization of qubits can be mimicked as illustrated in Example~\ref{exp:errors}. Phase flip decoherence errors can be mimicked in a similar fashion, by applying a Z-operation to the qubit. 
The amplitude damping error cannot be simulated using simple gate operations. This is due to the fact that damping a qubit is not reversible, which is why it cannot be expressed using reversible (unitary) operations. This makes it so, that the probability of applying the error is influenced by the state the error is applied to, as illustrated in the following example.
\begin{example}
\label{exp:vec_amp_damping}
Consider the state  $\ket{\psi^\prime} = \frac{1}{\sqrt{2}}(\ket{00} + \ket{11})$ from Example~\ref{exp:matrxi_vector_mul} (shown in Fig~\ref{fig:statevectordd}) and suppose that it is subject to amplitude damping. Here, things get more interesting. If amplitude damping affects the first qubit with probability $p$, its action is given by the matrices $\textup{A}_{0}=\begin{bNiceMatrix}[r][small]0&\sqrt{p}\\0&0\end{bNiceMatrix}$ and $\textup{A}_{1}=\begin{bNiceMatrix}[r][small]1&0\\0&\sqrt{1-p}\end{bNiceMatrix}$~\cite{NC:2000}. But, it is not the error probability $p$ (alone) that matters. In contrast to depolarizing and phase flip errors, amplitude damping is manifestly state-dependent. 
In order to get the probability for applying either $\textup{A}_{0}$ or $\textup{A}_{1}$, they have to be applied to $\ket{\psi^\prime}$. 
Applying $\textup{A}_{0}$ to $\ket{\psi^\prime}$ results in a state vector whose squared norm is $\tfrac{p}{2}$, which is also the probability that $\textup{A}_{0}$ is applied. Analogously, the probability for applying $\textup{A}_{1}$ is $1-\tfrac{p}{2}$. Depending on those probabilities, one state is randomly chosen and normalized, while the other one is discarded.
Thus damping the first qubit with probability $p$ results in the ensemble $\{(\tfrac{p}{2},\ket{01}), (1-\tfrac{p}{2}, \frac{1}{\sqrt{2-p}}\ket{00} + \frac{\sqrt{1-p}}{\sqrt{2-p}}\ket{11})\}$ (their decision diagram representations are also given in Fig.~\ref{fig:applyDepolToDD}).
\end{example}
 
Using all that, the concepts for stochastic quantum circuit simulation as reviewed in Section~\ref{sec:stochastic_simulation} can be realized by means of decision diagrams.
More precisely, recall from Section~\ref{sec:background} that for simulating quantum circuits we need means to represent vectors and matrices for states and operations, respectively. Additionally, we need some means of applying operations to states (either for applying quantum operations or error effects). Having all that, the stochastic approach presented in Section~\ref{sec:stochastic_simulation}---which allows to apply error operations probabilistically---can be used in a straightforward fashion. Since decision diagrams often allow to represent all these entities and to conduct all these operations in a much more compact and efficient fashion, a big challenge of existing approaches for stochastic quantum circuit simulation is addressed.

\subsection{Exploiting Concurrency Across Different Simulation Runs}

As detailed above, decision diagrams provide a powerful data structure that often helps to escape exponential memory requirements. This facilitates the faithful execution of \mbox{moderate-scale} quantum simulation. At the same time, however, decision diagrams can hardly exploit concurrency thus far~\cite{Hillmich2020}. 
This is in stark contrast to state-of-the art quantum simulators (such as~\cite{forest,qiskit,atos2016,qxSimulator2017,DBLP:journals/corr/WeckerS14,CirqPythonFramework,jones2018quest,DBLP:journals/corr/SmelyanskiySA16,villalonga2019highperformancesimulator,Steiger2018projectqopensource}) which heavily make use of concurrent executions (e.g., during matrix-vector multiplication). %

Stochastic quantum simulation, however, is an interesting use case where the apparent trade-off between optimizing memory (through decision diagrams) and exploiting concurrency (to accelerate matrix-vector multiplication) %
 can be resolved by different means:
Simply use different cores to compute \emph{independent} simulation runs (samples). 
Given that accurate predictions require a sufficient number of independent samples (see Theorem~\ref{theorem_stoch}),
the potential of exploiting concurrency across different simulation runs---rather than within individual runs---is enormous.

It is worthwhile to point out that concurrent execution is a general  feature of Monte-Carlo-type approximations and well known. Stochastic quantum simulation is merely an interesting special case.
Implementations of 
stochastic quantum simulations, e.g., in~\cite{atos2016,qiskit,forest}, do not seem to utilize this potential yet (most likely, because most existing approaches rely on exponentially large vector and matrix representations limiting the potential of having several runs of this size in parallel and, hence, exploiting concurrency during the matrix-vector multiplications seemed to be the more feasible approach). With the proposed approach, both (memory-efficient representations \emph{and} concurrent executions) can be exploited.

\section{Evaluation}
\label{sec:evaluation}
In order to empirically evaluate the performance of the proposed stochastic error simulation approach, we implemented the concepts described above in C++ (using the open-source decision diagram package taken from~\cite{zulehner2019package}).
Afterwards, we compared the resulting performance against other available \mbox{state-of-the-art} stochastic simulators, namely the \emph{LinAlg} simulator from the \emph{Atos Quantum Learning Machine} (QLM)~\cite{atos2016} and the \emph{statevector} simulator from IBM's \emph{Qiskit}~\cite{qiskit}.

We considered different benchmark sets: First, we evaluated all simulation approaches using the \emph{Entanglement} circuit (an algorithm generating the GHZ state), as well as the \emph{Quantum Fourier Transform} (QFT,~\cite{NC:2000}), with an increasing number of qubits.
By this, we considered typical use cases incorporating quantum-mechanical effects such as superposition and entanglement in a scalable fashion (i.e., with an increasing number of qubits).
Second, we evaluated all simulation approaches using the circuits from the benchmark suite \emph{QASMBench} (taken from~\cite{Li2020}), which contains a broad range of different quantum algorithms. %

For all benchmarks, we considered all errors discussed in Section~\ref{subsec:back_errors}, i.e., gate errors, as well as decoherence errors. More precisely, we applied a depolarization error with 0.1~\% probability, an amplitude damping (T1) error with 0.2~\% probability, and a phase flip error (T2) with 0.1~\% probability to the gate/qubit. The errors have been simulated using the stochastic approach presented in Section~\ref{sec:stochastic_simulation} with a total of \mbox{$M=30,000$} iterations for each benchmark (using Theorem~\ref{theorem_stoch}, this corresponds to tracking 1000 properties with an error margin of $<0.01$ and a confidence of 95~\%).

\begin{table*}[ht]
	\caption{Evaluation results}
	\label{tab:results}
	\centering	
	\begin{subfigure}[b]{0.30\linewidth}
		\centering
		\caption{Entanglement circuits}
		\label{tab:results_ent}
		\begin{tabular}{r|r|r|r}
		$n$ & Qiskit [s] & QLM [s] & Proposed [s] \\\hline\hline		
		\begin{filecontents*}{csv/results_ent.csv}
qubit,qisent,qlment,jkuent
21,1075.74,22.05,\textbf{0.58}
22,2247.9,42.57,\textbf{0.88}
23,\textgreater 3600,155.22,\textbf{0.68}
$\smash{\vdots}$,\textgreater 3600,\smash{\vdots},\smash{\vdots}
27,\textgreater 3600,1186.58,\textbf{0.78}
28,\textgreater 3600,2251.19,\textbf{0.89}
29,\textgreater 3600,\textgreater 3600,\textbf{0.88}
$\smash{\vdots}$,\smash{\vdots},\smash{\vdots},\smash{\vdots}
63,\textgreater 3600,\textgreater 3600,\textbf{2.61}
64,\textgreater 3600,\textgreater 3600,\textbf{2.72}
\end{filecontents*}
\csvreader[
		late after line=\\,
		]{csv/results_ent.csv}	
		{1=\qubitent, 2=\qisent, 3=\qlment, 4=\jkuent}
		{\qubitent & \qisent & \qlment & \jkuent}		
		\end{tabular}		
	\end{subfigure}	\hfill
		\begin{subfigure}[b]{0.30\linewidth}
		\centering
		\caption{QFT circuits}
		\label{tab:results_qft}
		\begin{tabular}{r|r|r|r}
		$n$ & Qiskit [s] & QLM [s] & Proposed [s] \\\hline\hline		
		\begin{filecontents*}{csv/results_qft.csv}
qubit,qisqft,qlm,jkuqft
12,20.61,1679.412,\textbf{2.41}
13,31.07,2338.356,\textbf{3.34}
14,60.13,\textgreater 3600,\textbf{4.25}
$\smash{\vdots}$,\smash{\vdots},\smash{\vdots},\smash{\vdots}
17,834.95,\textgreater 3600,\textbf{7.22}
18,2115.71,\textgreater 3600,\textbf{8.97}
19,\textgreater 3600,\textgreater 3600,\textbf{10.17}
$\smash{\vdots}$,\smash{\vdots},\smash{\vdots},\smash{\vdots}
63,\textgreater 3600,\textgreater 3600,\textbf{377.68}
64,\textgreater 3600,\textgreater 3600,\textbf{402.65}
\end{filecontents*}
\csvreader[
		late after line=\\,
		]{csv/results_qft.csv}	
		{1=\qubitqft, 2=\qisqft, 3=\qlmqft, 4=\jkuqft}
		{\qubitqft & \qisqft & \qlmqft & \jkuqft}		
		\end{tabular}		
	\end{subfigure}	\hfill
	\begin{subfigure}[b]{0.36\linewidth}
		\centering
		\caption{QASMBench circuits}
		\label{tab:results_qasm}
		\begin{tabular}{l|r|r|r}
		Name & $n$ & Qiskit [s] & Proposed [s] \\\hline\hline		
		\begin{filecontents*}{csv/results_qasm.csv}
name,qubit,qisother,jkuother
basis\_trotter,4,112.69,\textbf{47.85}
vqe\_uccsd,6,181.82,\textbf{88.59}
vqe\_uccsd,8,\textbf{520.37},3600.00
ising,10,\textbf{20.48},282.23
seca,11,12.32,\textbf{1.29}
sat,11,37.33,\textbf{2.12}
multipler,15,136.83,\textbf{0.78}
bigadder,18,829.33,\textbf{0.98}
cc,18,\textbf{236.60},3600.00
bv,19,448.14,\textbf{211.35}
\end{filecontents*}
\csvreader[
		late after line=\\,
		]{csv/results_qasm.csv}	
		{1=\name, 2=\qubitother, 3=\qisother, 4=\jkuother}
		{\name & \qubitother & \qisother & \jkuother}		
		\end{tabular}		
	\end{subfigure}
\end{table*}

Table~\ref{tab:results} summarizes the results of our evaluation. 
More precisely, Table~\ref{tab:results_ent}, Table~\ref{tab:results_qft}, and Table~\ref{tab:results_qasm} provide the results for the entanglement benchmark, the QFT benchmark, and the QASMBench benchmarks, respectively.
In each table, we provide the number~$n$ of qubits, as well as the required runtime for each simulation approach in seconds. %
Note that, due to space limitations, only a selection of the QASMBench benchmarks is explicitly listed. Out of the 53 benchmarks, 4 have been omitted as they could not be simulated by either simulation approach within the time limit of 1~hour. 39 benchmarks have been omitted because their differences in runtime between the simulation approaches remained rather small.  %
Furthermore, note that we do not list any results from Atos' QLM simulator for the QASMBench benchmarks, since those circuits are only provided in the OpenQASM format, which is not supported by the QLM simulator. 

The results %
show the improved performance of the proposed simulation approach compared to the \mbox{state-of-the-art} approaches by Atos and IBM for certain algorithms. For the entanglement and QFT benchmarks, a substantially better scalability with respect to the number of qubits can be reported. For the QASMBench benchmarks, the proposed solution reaches its limit and gives longer runtimes for the \emph{ising}, \emph{yqe\_uccsd}, and \emph{cc} circuits. In the other cases, the circuits could be simulated faster---some algorithms even by several orders of magnitude.
\section{Conclusion}
\label{sec:conclusion}

Quantum circuit simulation is an important research area. However, many available quantum circuit simulators simplify the problem by simulating \emph{perfect} quantum computers. Due to the fragile nature of quantum systems, quantum computers are always subject to errors during their calculations. Simulators which allow the consideration of errors during the simulation often suffer from the exponential complexity of vectors and matrices required for the simulation.
We addressed this issue by implementing a stochastic circuit simulator, which exploits the compact representations of vectors and matrices offered by decision diagrams and utilizes concurrent executions for an efficient generation of samples. 
Evaluations and comparisons against state-of-the-art simulators by IBM and Atos show the improved scalability and efficiency of the proposed solution for certain applications.

\section*{Acknowledgments}
This work has partially been supported by the University of Applied Sciences PhD program of the State of Upper Austria (managed by the FFG), by the LIT Secure and Correct Systems Lab funded by the State of Upper Austria, as well as by the BMK, BMDW, and the State of Upper Austria in the frame of the COMET program (managed by the FFG).

\bibliographystyle{IEEEtran}
\let\oldbibliography\thebibliography
\renewcommand{\thebibliography}[1]{%
  \oldbibliography{#1}%
  \setlength{\itemsep}{0pt plus 3pt}%
}
\bibliography{../../bib/lit_header, IEEEabrv,../../bib/lit_quantum, ../../bib/lit_myrev,../../bib/lit_misc,../../bib/lit_mymisc,../../bib/lit_others,../../bib/lit_othersrev,../../bib/lit_rev,../../bib/lit_adiabatic,../../bib/lit_memristor,../new-references} 
\end{document}